\documentclass[journal]{IEEEtran}
\usepackage[cmex10]{amsmath}
\usepackage{amssymb,amsthm,amsfonts,mathrsfs}
\usepackage{graphicx}
\usepackage{dsfont,enumitem,stackrel}
\usepackage{color}
\usepackage{mathtools,amssymb,bm}
\usepackage{amstext}
\usepackage{array}
\usepackage[ruled]{algorithm}
\usepackage{algorithmic}
\usepackage{cases}
\usepackage{pgfplots}
\usepackage{accents}
\usepackage{caption}
\usepackage{caption}
\DeclareCaptionLabelSeparator{periodspace}{.\quad}
\usepackage{float}
\usepackage{cite}
\usepackage [autostyle, english = american]{csquotes}
\usepackage{hyperref}
\theoremstyle{remark}

\newcommand\ASTART{\bigskip\noindent\begin{minipage}[b]{0.5\linewidth}}
	
	\newcommand\AENDSKIP{\end{minipage}\bigskip}
\newcommand\AEND{\end{minipage}}
\ifCLASSOPTIONcaptionsoff
\usepackage[nomarkers]{endfloat}
\let\MYoriglatexcaption\caption
\renewcommand{\caption}[2][\relax]{\MYoriglatexcaption[#2]{#2}}
\fi
\def\R{\mathbb{R}}

\theoremstyle{plain}
\newtheorem{thm}{\textbf{Theorem}}

\theoremstyle{definition}

\theoremstyle{remark}

\newcommand*{\rom}[1]{\expandafter\@slowromancap\romannumeral #1@}

\DeclarePairedDelimiter\floor{\lfloor}{\rfloor}
\newcommand{\norm}[1]{\left\|#1\right\|}

\usepackage{standalone}
\usepackage{tikz}
\usetikzlibrary{shapes,arrows}
\usepackage{subcaption}
\graphicspath{{./Figures/}} 
\begin{document}
%
%\onecolumn
\title{Adaptive Recovery of Dictionary-sparse Signals using Binary Measurements}
\author{Hossein~Beheshti, Sajad~Daei, Farzan~Haddadi
\thanks{H. Beheshti, S. Daei and F. Haddadi are with the School of 
Electrical Engineering, Iran University of Science \& Technology.}
}%
\maketitle

\begin{abstract}
One-bit compressive sensing (CS) is an advanced version 
of sparse recovery in which the sparse signal of interest 
can be recovered from extremely quantized measurements. 
Namely, only the sign of each measurement is available to us. 
In many applications, the ground-truth signal is not 
sparse itself, but can be represented in a redundant 
dictionary. A strong line of research has addressed 
conventional CS in this signal model including its 
extension to one-bit measurements. However, one-bit 
CS suffers from the extremely large number of required 
measurements to achieve a predefined reconstruction error level. 
A common alternative to resolve this issue is to exploit adaptive 
schemes. Adaptive sampling acts on the acquired samples to 
trace the signal in an efficient way. In this work, we 
utilize an adaptive sampling strategy to recover 
dictionary-sparse signals from binary measurements. For this task, 
a multi-dimensional threshold is proposed to incorporate the 
previous signal estimates into the current sampling procedure. 
This strategy substantially reduces the required number of 
measurements for exact recovery. Our proof approach is 
based on the recent tools in high dimensional geometry 
in particular random hyperplane tessellation and Gaussian width. 
We show through rigorous and numerical analysis that 
the proposed algorithm considerably outperforms state of 
the art approaches. Further, our algorithm reaches an 
exponential error decay in terms of the number of 
quantized measurements.
\end{abstract}


% \fi
\IEEEpeerreviewmaketitle

\section{Introduction}
\IEEEPARstart{S}{ampling} a signal heavily depends 
on the prior information about the signal structure. 
For example, if one knows the signal of interest is 
band-limited, the Nyquist sampling rate is sufficient 
for exact recovery. Signals with most of its coefficients 
zero are called sparse. It has been turned out that sparsity is a 
powerful assumption that results in a significant reduction 
in the required number of measurements. The process of recovering a sparse signal from a few number of measurements is called compressed sensing (CS). In CS, the measurement vector is assumed to be a linear combination of the ground-truth signal i.e.,
\begin{align}
\label{eq:cs}
\bm{y} = \bm{A}\bm{x},
\end{align} 
where $ \bm{A}\in \R^{m\times N} $ is called the measurement matrix, whose random elements are drawn from a normal distribution and 
$ \bm{x}\in \R^{N} $ is an unknown $ s $-sparse 
signal i.e. it has at most $ s $ nonzero entries, i.e., $\| \bm{x}\|_{0} \leq s$.
Here, $\norm{\cdot}_{0}$ is the $ \ell_{0} $ norm which 
counts the number of non-zero elements. It is shown 
that $\mathcal{O}(s\log(\tfrac{N}{s})) $ measurements is sufficient to 
guarantee exact recovery of the signal,  
by solving the convex program:
\begin{align}
	\text P_{1}: \quad \min_{{\bm{z}}\in \R^{N} } 
	\quad \norm{{\bm{z}}}_{1} \quad \text{s.t.} \quad \bm{y}= \bm{A}\bm{z},
\end{align}
with high probability (see \cite{donoho2006compressed},\cite{candes2006compressive}).

Practical limitations enforce us to quantize the measurements 
in \eqref{eq:cs} as $\bm{y}=\mathcal{Q}(\bm{A x})$ where 
$\mathcal{Q}:\mathbb{R}^m\rightarrow \mathcal{A}^m$ is a 
non-linear operator that maps the measurements into a 
finite symbol alphabet $\mathcal{A}$. It is an interesting 
question that what is the result of extreme quantization? 
This question is addressed in \cite{boufounos20081} which 
states that signal reconstruction is still feasible using 
only one-bit quantized measurements. 

In one-bit compressed sensing, samples are taken 
as the sign of a linear transform of the signal $
\bm{y}= \text{sign}\left(\bm{A}\bm{x}\right)$.
This sampling scheme discards magnitude information. 
Therefore, we can only recover the direction of the signal. 
Fortunately, changing the threshold randomly in each 
measurement as $\tau_i \sim \mathcal N (0,1)$ conserves 
the amplitude information. Thus, the new sampling scheme reads $\bm{y}= \text{sign}\left(\bm{A}\bm{x}-\bm{\tau}\right)$.
While a great part of CS literature discusses 
sparse signals, most of the natural signals
are dictionary-sparse i.e. sparse in a transform domain. 
For instance, sinusoidal signals and natural image are 
sparse in Fourier and wavelet domains, respectively \cite{Candes2011}. 
This means that our signal of interest $\bm{f}\in\mathbb{R}^n$ 
can be expressed as $\bm{f}=\bm{D x}$ where $\bm{D}\in\mathbb{R}^{n\times N}$ is a 
redundant dictionary and $\bm{x}\in\mathbb{R}^N$ is a sparse vector. 
A common approach for recovering such signals, is to use the 
optimization problem 
\begin{align}
\text P_{1,\bm{D}}: \quad \min_{\bm{z}\in \R^{N}} \quad 
\norm{\bm{D}^{\rm H}\bm{z}}_{1} \quad \text{s.t.} \quad \bm{y}= \bm{A}\bm{z},
\end{align}
which is called $\ell_1$ analysis problem. 

In this work, we investigate a more practical situation where the signal of 
interest $\bm{f}$ is {\it effective} $s$-analysis-sparse which means that $\bm{f}$ satisfies $\|\bm{D}^{\rm H}\bm{f}\|_1\le \sqrt{s}\|\bm{D}^{\rm H}\bm{f}\|_2$.
In fact, perfect dictionary-sparsity is rarely satisfied in practice, 
since real-world signals of interest are only compressible in a domain. 
Our approach is adaptive which means that we incorporate previous 
signal estimates into the current sampling procedure. 
More explicitly, we solve the optimization problem
\begin{align}\label{eq:onebit+analysis}
 \min_{\bm{z}\in \R^{N}} ~ \norm{\bm{D}^{\rm H}\bm{z}}_{1} ~
 \text{s.t.} ~ \bm{y}:={\rm sign}(\bm{A}\bm{f}-\bm{\varphi})={\rm sign}(\bm{A}\bm{z}-\bm{\varphi}),
\end{align}
where $\bm{\varphi}\in\mathbb{R}^m$ is chosen adaptively 
based on previous estimations.
We propose a strategy to 
find a best effective $s$-analysis-sparse approximation to a signal in $\mathbb{R}^n$.
\subsection{Contributions}\label{subsection.Contribution}
In this section, we state our novelties in compared with
previous works. To highlight the contributions, we list them as below.
\begin{enumerate}
	\item \textbf{Proposing a novel algorithm for dictionary-sparse signals}: We 
	introduce an adaptive thresholding algorithm for reconstructing 
	dictionary-sparse signals in case of binary measurements. 
	The proposed algorithm provides accurate signal estimation even 
	in case of redundant and coherent dictionaries. The required 
	number of one-bit measurements considerably outperforms 
	the non-adaptive approach used in \cite{Baraniuk2017}.
	\item \textbf{Exponential decay of reconstruction error}: The error 
	of our algorithm exhibits exponential decaying behavior as 
	long as the number of adaptive stages sufficiently grows. 
	To be more precise, we obtain a near-optimal relation between 
	the reconstruction error and the required number of adaptive 
	stages. Written in mathematical form, if one takes the output 
	of our reconstruction algorithm by $\widehat{\bm{f}}$, 
	then, we show that $\|\widehat{\bm{f}}-\bm{f}\|_2\approx \mathcal{O}(\tfrac{1}{2^T})$,
	where $\bm{f}$ is the ground-truth signal and $T$ is 
	the number of stages in our adaptive algorithm 
	(see \ref{thm.main} for more details)
	\item \textbf{High dimensional threshold selection}: We propose an 
	adaptive high-dimensional threshold to extract the most information 
	from each sample, which substantially improves performance and 
	reduces the reconstruction error (see 
	\ref{subsection.HDT} for more explanations).
\end{enumerate}
\subsection{Prior Works and Key Differences}\label{subsection.PriorWork}
In this section, we review prior works about applying 
quantized measurements to CS framework\cite{boufounos20081,zymnis2010compressed,laska2012regime,jacques2013robust,plan2013one,kamilov2012one,baraniuk2017exponential,Baraniuk2017}. In what follows, we explain some of them.

The authors of 
\cite{boufounos20081} propose a heuristic algorithm 
to reconstruct the ground-truth sparse signal from 
extreme quantized measurements i.e. one bit measurements. In \cite{zymnis2010compressed}, it has been shown that 
conventional CS algorithms also works well when the measurements 
are quantized. In \cite{jacques2013robust} an algorithm 
with simple implementation is proposed. This algorithm posses 
less error in terms of hamming distance compared with the existing 
algorithms. Investigated from a geometric view, the authors 
of \cite{plan2013one}, exploit functional analysis tools to 
provide an almost optimal solution to the problem of one-bit CS.
They show that the number of required one-bit measurements 
is $\mathcal{O}(s\log^2(\tfrac{n}{s}))$. 
In a different approach, the work 
\cite{baraniuk2017exponential} proposes an adaptive quantization 
and recovery scenario making an exponential error decay in 
one-bit CS frameworks. Many of the techniques mentioned for 
adaptive sparse signal recovery do not generalize (at least 
not in an obvious strategy) to dictionary-sparse signal. For example, determining a surrogate of $\bm{f}$ that is supposed to be of lower-complexity with respect to $\bm{D}^{\rm H}$ is non-trivial and challenging. We should emphasize that, while the proofs and main parts in \cite{baraniuk2017exponential} relies on hard thresholding operator, it could not be used for either effective or exact dictionary sparse signals. This is due to that given a vector $\bm{x}$ in the analysis domain $\mathbb{R}^N$, one can not guarantee the existence of a signal $\bm{f}$ in $\mathbb{R}^n$ such that $\bm{D}^{\rm H}\bm{f}=\bm{x}$. 
Recently the work \cite{Baraniuk2017} shows, both direction 
and magnitude of a dictionary-sparse signal can be recovered 
by a convex program with strong guarantees. The 
work \cite{Baraniuk2017} has inspired our work for recovering 
dictionary-sparse signal in an adaptive manner.
In contrast to the existing 
method \cite{Baraniuk2017} for binary dictionary-sparse signal 
recovery which takes all of the measurements in one step with 
fixed settings, we solve the problem in an adaptive multi-stage way. 
In each stage, regarding the estimate from previous stage, 
our algorithm is propelled to the desired signal.
In the non-adaptive work \cite{Baraniuk2017}, the error rate 
is poorly large while in our work, the error rate 
exponentially decays with increasing the number of 
adaptive steps.

\textit{Notation}.\label{subsection.notation}
Here we introduce the notation used in the paper.
Vectors and matrices are denoted by boldface lowercase 
and capital letters, respectively.  
$ \bm{v}^T $ and $ \bm{v}^{\rm H} $ stand for 
transposition and hermitian of $ \bm{v} $, 
respectively. $ C$, and $c $ denote 
positive absolute constants which can be different 
from line to line. We use $ \norm{\bm{v}}_2= \sqrt{\sum_i |v_i|^2} $ 
for the $ \ell_2 $-norm of a vector $ \bm{v} $ in $ \R^n $, 
$ \norm{\bm{v}}_1= \sum_{i} |v_i|$ for the $\ell_1$-norm 
and $ \norm{\bm{v}}_{\infty}= \max_{i} |v_{i}| $ for 
the $ \ell_{\infty} $-norm. We write $ \mathbb{S}^{n-1}:= \{\bm{v}\in \R^n : \norm{\bm{v}}_{2}= 1 \} $   
for the unit Euclidean sphere in $ \R^n $.
For $\bm{x}\in \mathbb{R}^{n}$, we define $\bm{x}_S$ as the sub-vector
in $\mathbb{R}^{|S|}$ consisting of the entries indexed by the set $S$.
\begin{figure}[t!]
	\centering
	\tikzstyle{block} = [draw, fill=blue!20, rectangle, 
minimum height=3em, minimum width=6em]
\tikzstyle{sum} = [draw, fill=blue!20, circle, node distance=2cm]
\tikzstyle{input} = [coordinate]
\tikzstyle{output} = [coordinate]
\tikzstyle{pinstyle} = [pin edge={to-,thin,black}]
\begin{tikzpicture}[auto, node distance=2cm,>=latex']
% We start by placing the blocks
\node [input, name=input] {};

\node [block, right of=input,node distance=3cm] (CoeS) {
	$\begin{array}{c}
	\bm{x} \\
	\bm{D}^{\text{H}}\bm{f} 
	\end{array}	$};

\node [block, right of=CoeS,node distance=4cm] (SigS) {$ \bm{f} $};

\node [block, right of=SigS,node distance=3cm] (system) {$ \bm{A}^{(i)}\bm{f} $};

\node [sum, right of=system] (sum) {};
    
\node [block, right of=sum,node distance=3cm] (Sign) { $ \text{sign}\left(\bm{A}^{(i)}\bm{f}- \bm{\varphi}^{(i)}\right) $};

\node [block, below of=Sign,node distance=1.5cm] (ATh) { $ \bm{\varphi}^{(i)} $};

\node [output, right of=Sign] (yTemp) {};
\node [output, right of=yTemp] (yOut) {};

% We draw an edge between the controller and system block to 
% calculate the coordinate u. We need it to place the measurement block. 
\draw [->] (CoeS) -- node {$\rightarrow\bm{D}$} (SigS);

\draw [->] (SigS) -- node {$\bm{D}^{\text{H}}\leftarrow$} (CoeS);

\draw [->] (SigS) -- node {$\bm{A}$} (system);

\draw [draw,->] (system) -- node {$+$} (sum);

\draw [draw,->] (sum) -- node {} (Sign);

\draw [draw,-] (Sign) -- node {} (yTemp);

\draw [draw,->] (yTemp) -- node {$ \bm{y} $} (yOut);

\draw [->] (yTemp) |- (ATh);

\draw [->] (ATh) -| node[pos=0.99] {$-$} node [near end] {} (sum);

%\draw [->] (controller) -- node[name=u] {$u$} (system);
\end{tikzpicture}
	\caption{Block diagram of adaptive sampling procedure.}
	\label{fig:MeasureBD}
\end{figure}
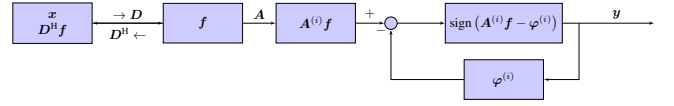
\section{Main result}\label{section.MainResult}
Our system model is built upon the optimization problem 
(\ref{eq:onebit+analysis}). A major part of this problem 
is to choose an efficient threshold $\bm{\varphi}\in\mathbb{R}^m$. 
To this end, we propose a closed-loop feedback system 
(see Figure \ref{fig:MeasureBD}) which exploits prior information 
from previous stages. Our adaptive approach consists of three algorithms $\mathsf{HDTG}$ (Algorithm \ref{alg:HDTG}), $\mathsf{AS}$ (Algorithm \ref{alg:AQ}), and $\mathsf{AR}$ (Algorithm \ref{alg:AR}). $\mathsf{HDTG}$ specifies the method of choosing the thresholds based on previous estimates. $\mathsf{AS}$ takes one-bit measurements of the signal by using $\mathsf{HDTG}$ and returns the one-bit samples and corresponding thresholds. Lastly, $\mathsf{AR}$ recovers the signal given the one bit measurements.
We provide a mathematical 
framework to guarantee our algorithm results in the following. 
\begin{thm}[Main theorem]\label{thm.main}
	Let $\bm{f}\in \mathbb{R}^n$ be the desired effective $s$-analysis-sparse
	signal with $\|\bm f\|_2\le r$ and $ \bm{A}\in \R^{m\times n} $ be 
	the measurement matrix with standard normal entries where $m$ is the total number 
	of measurements divided into $T$ stages. Assume that 
	\begin{align*}
	[\bm{y}, \bm{\varphi}]&=\mathsf{AS}(\bm{A},\bm{f},\bm{D},r,T),\\
	\widehat{\bm{f}}&=\mathsf{AR}(\bm{A},\bm{D}, \bm{y}, \bm{\varphi}, r,T)
	\end{align*}
	be the sampling and recovery algorithms introduced later in Algorithms \ref{alg:AQ} 
	and \ref{alg:AR}, respectively where $\bm{\varphi}$ is determined by Algorithm \ref{alg:HDTG}.
	Then, with probability at least $1- \gamma \exp(\tfrac{-cm}{T})$ over the 
	choice of $\bm{\varphi}$ and $\bm{A}$,
	the output of Algorithm \ref{alg:AR} satisfies
	\begin{align}
	\|\bm{f}-\hat{\bm{f}}\|_{2}\leq \epsilon r2^{1-T}.
	\end{align}
\end{thm}
\begin{proof}
	See Appendix \ref{proof.theoremmain}.
\end{proof}

A remarkable note is that if we only consider one stage 
i.e. $T=1$, the exponential behavior of our error bound 
disappears and reaches the state of art error bound  
\cite[Theorem 8]{Baraniuk2017}. In fact, the results of 
\cite{Baraniuk2017} is a special case of our work when the 
thresholds are non-adaptive. The term \textit{adaptivity} is related 
to the threshold updating and the measurement matrix $\bm{A}$ is fixed.

In what follows, we provide rigorous explanations about our proposed 
algorithms \ref{alg:HDTG}, \ref{alg:AQ} and \ref{alg:AR} in three items.
\begin{enumerate}
	
\item \label{subsection.HDT}\textbf{High Dimensional Threshold}:
Our algorithm for high dimensional threshold selection is given in Algorithm \ref{alg:HDTG}. 
The algorithm output consists of two parts: deterministic (i.e., $\bm{A}\widehat{\bm{f}}$) and random (i.e., $\bm{\tau}\in\mathbb{R}^q$) . The former transfers the origin to the previous signal estimate ($\widehat{\bm{f}}$) while the latter creates measurement dithers from the origin. (this dither is 
controlled by the variance parameter $\sigma^2$).
\begin{algorithm}[!h]
	\caption{$ \mathsf{HDTG} $: High dimension threshold generator}
	\label{alg:HDTG}
	\begin{algorithmic}[1]
		\renewcommand{\algorithmicrequire}{\textbf{Input:}}
		\renewcommand{\algorithmicensure}{\textbf{Output:}}
		\REQUIRE Mapping matrix $ \bm{A} $, 
		number of measurements $ q $, dithers variance $ \sigma^{2} $, signal estimation $ \widehat{\bm{f}} $.
		\ENSURE High dimension threshold vector $\bm{\varphi}\in \R^{q}$.
		\STATE $ \bm{\tau}\sim N(0,\sigma^{2}\bm{I}_{q} ) $
		\STATE  $ \bm{\varphi}=\bm{A}\widehat{\bm{f}}+\bm{\tau} $
		\RETURN $ \bm{\varphi} $
	\end{algorithmic} 
\end{algorithm}
\item \textbf{Adaptive Sampling}:
Our adaptive sampling algorithm is given in Algorithm \ref{alg:AQ}. 
To implement this algorithm, we need the dictionary 
$ \bm{D}\in\mathbb{R}^{n\times N} $,  the measurement 
matrix $ \bm{A}\in\mathbb{R}^{m\times n} $,  linear measurement 
$ \bm{A}\bm{f} $ and an over estimation of signal power 
$ r $ ($ \norm{\bm{f}}_{2}\leq r $). At the first stage, 
we initialize signal estimation to zero vector. We choose $T$ stages for our algorithm. 
At the $i~$th stage, the measurement matrix $\bm{A}^{(i)}$ is taken from 
the $i~$th row subset (of size $q:=\floor*{\tfrac{m}{T}}$) of $\bm{A}$. 
The adaptive sampling process consists 
of four essential parts. First, in step $2$ of the pseudo code, we generate the high dimensional thresholds 
using Algorithm \ref{alg:HDTG} by the parameters $\sigma^2=2^{1-i}r$
 and $\widehat{\bm{f}}_i$. 
Second, we compare the linear measurement block $\bm{A}^{(i)}\bm{f}$ with 
the threshold vector $\bm{\varphi}^{(i)}$ and obtain the sample 
vector $\bm{y}^{(i)}$ (step $3$ of Algorithm \ref{alg:AQ}). Third, we implement a second 
order cone program (steps $4$ and $5$) to find an approximate for $f$. However, this strategy does not often lead to an effective dictionary sparse signal. So, in step $6$, we devise a strategy to find a low-complexity approximation of $f$ with respect to operator $\bm{D}^{\rm H}$. In other words, we project the resulting signal onto the nearest effective $s$-analysis-sparse signal.
We refer to the third part of the adaptive sampling algorithm (steps 4-6) as single step recovery (SSR).
The estimated signal at each stage ($\bm{f}_{i}$) builds the deterministic part of our high dimensional 
threshold in step $2$. Finally, Algorithm \ref{alg:AQ} returns binary 
vectors $\{\bm{y}^{(i)}\}_{i=1}^T$ and the threshold vectors 
$\{\bm{\varphi}^{(i)}\}_{i=1}^T$ to the output.
\begin{algorithm}[!h]
	\caption{$ \mathsf{AS} $: Adaptive Sampling}
	\label{alg:AQ}
	\begin{algorithmic}[1]
		\renewcommand{\algorithmicrequire}{\textbf{Input:}}
		\renewcommand{\algorithmicensure}{\textbf{Output:}}
		\REQUIRE Dictionary $ \bm{D}\in\R^{n\times N} $, 
		measurement matrix $ \bm{A}\in\R^{m\times n} $, 
		linear measurement $ \bm{Af}\in \R^{m} $, 
		norm estimation $ \norm{\bm{f}}_{2}\leq r $, 
		number of blocks $ T $.
		\ENSURE  Quantized measurements $ \bm{y} \in \lbrace\pm 1\rbrace^{m} $, 
		high dimension thresholds $ \bm{\varphi} \in \R^{m} $
		\\ \textit{Initialization} : $ \bm{f}_{0}\leftarrow \bm{0} $, 
		$ q = \floor*{\tfrac{m}{T}}  $, $ \bm{A}^{(i)}\in \R^{q\times n}  $
		\FOR {$i = 1,\cdots,T$}
		\STATE $ \bm{\varphi}^{\left(i\right)}\leftarrow \mathsf{HDTG}(\bm{A}^{(i)},q,2^{1-i}r,\bm{f}_{i-1}) $
		\STATE $ \bm{y}^{\left(i\right)} = \text{sign}\left(\bm{A}^{(i)}\bm{f}-\bm{\varphi}^{\left(i\right)}\right)$
		\STATE 
		$ \bm{\Delta}_{i}\leftarrow \mathop{\arg\min}_{\bm{z}\in \R^{n} } \|\bm{D}^{\rm H}\bm{z}\|_{1}$ \\ 
		\quad $ ~s.t.~   \bm{y}^{\left(i\right)}=\text{sign}\left(\bm{A}^{\left(i\right)}\bm{z}-\bm{\varphi}^{\left(i\right)}\right), \: \norm{\bm{z}}_{2}\leq 2^{1-i}r $ 
		\STATE $ \bm{f}_{tmp} = \bm{f}_{i-1}+\bm{\Delta}_{i}$
		\STATE 
		$ \bm{f}_{i}\leftarrow \mathop{\arg\min}_{\bm{z}\in \R^{n} } \|\bm{z}-\bm{f}_{tmp}\|_{2}$  
		$ ~s.t.~   \: \norm{\bm{D}^{\rm H}\bm{z}}_{1}\leq \sqrt{s}r $ 
		\ENDFOR
		\RETURN $ \bm{y}^{\left(i\right)},\bm{\varphi}^{\left(i\right)} $ for $ i=1,\cdots,T $
	\end{algorithmic} 
\end{algorithm}
\begin{algorithm}[!h]
	\caption{$ \mathsf{AR}$: Adaptive Recovery}
	\label{alg:AR}
	\begin{algorithmic}[1]
		\renewcommand{\algorithmicrequire}{\textbf{Input:}}
		\renewcommand{\algorithmicensure}{\textbf{Output:}}
		\REQUIRE Dictionary $ \bm{D}\in\R^{n\times N} $, measurement matrix $ \bm{A}\in\R^{m\times n} $,
		quantized measurements $ \bm{y} \in \lbrace\pm 1\rbrace^{m} $,
		high dimension thresholds $ \bm{\varphi}\in\R^{m} $,
		norm estimation $ \norm{\bm{f}}_{2}\leq r $, 
		number of blocks $ T $.
		\ENSURE  Estimated vector $ \hat{\bm{f}}\in \R^{n} $.
		\\ \textit{Initialization} :  $ \bm{f}_{0}\leftarrow \bm{0} $,
		$ q = \floor*{\tfrac{m}{L}}  $, $ \bm{A}^{(i)}\in \R^{q\times n}$,
		$ \bm{y}^{\left(i\right)}\in \lbrace\pm 1\rbrace^{q}$, 
		$\bm{\varphi}^{\left(i\right)}\in \R^{q} $
		\FOR {$i = 1,\cdots,T$}
		\STATE
		$ \bm{\Delta}_{i}\leftarrow \mathop{\arg\min}_{\bm{z}\in \R^{n} } \norm{\bm{D}^{\rm H}\bm{z}}_{1}$ \\
		\quad $ ~s.t.~   \bm{y}^{\left(i\right)}=\text{sign}\left(\bm{A}^{\left(i\right)}\bm{z}-\bm{\varphi}^{\left(i\right)}\right), \: \norm{\bm{z}}_{2}\leq 2^{1-i}r $ 
		\STATE $ \bm{f}_{tmp} = \bm{f}_{i-1}+\bm{\Delta}_{i}$
		\STATE 
		$ \bm{f}_{i}\leftarrow \mathop{\arg\min}_{\bm{z}\in \R^{n} } \|\bm{z}-\bm{f}_{tmp}\|_{2}$  
		$ ~s.t.~   \: \norm{\bm{D}^{\rm H}\bm{z}}_{1}\leq \sqrt{s}r $ 
		\ENDFOR
		\RETURN $ \bm{f}_{T} $
	\end{algorithmic} 
\end{algorithm}
\item \textbf{Adaptive Recovery}:
In the recovery procedure (Algorithm \ref{alg:AR}), we need 
the dictionary $ \bm{D} $, the measurement matrix $ \bm{A} $, 
binary measurements vector $ \bm{y} $, high dimensional threshold 
vector $ \bm{\varphi} $ and an upper norm estimation of signal $ r $. 
In the adaptive recovery algorithm, we first divide the inputs ($\bm{y}$ 
and $\bm{A}$) into $ T $ blocks (i.e., $\bm{y}^{(i)}$ and $\bm{A}^{(i)}$). 
Then, we simply implement SSR on each block. The output of SSR in the last stage is
the final recovered signal (i.e., $\bm{f}_T$).
\end{enumerate}
\section{Numerical Experiments}\label{section.sim}
In this section, we investigate the performance of our algorithm 
and compare it with the two previous 
one-bit dictionary-sparse recovery given by \cite{Baraniuk2017}. 
The first algorithm solves linear programming optimization 
(LP) \cite[Subsection ~4.1]{Baraniuk2017} and the second 
algorithm solves a second-order cone programming (CP) optimization \cite[Subsection ~4.2]{Baraniuk2017}.
First, we construct a matrix where its columns 
are drawn randomly and independently from $ \mathbb{S}^{n-1} $. Then, the dictionary $ \bm{D}\in \R^{n\times N} $ ($ N=1000,n=50 $) is defined as an orthonormal basis of this matrix. Then, the signal $\bm{f}$ is generated as $\bm{f}=\bm{B}\bm{c}$ where $\bm{B}$ is a basis of ${\rm null}(\bm{D}_{\overline{\mathcal{S}}})$ and $c$ is drawn from standard normal distribution. Here, $\overline{\mathcal{S}}$ is used to denote the complement of the support of $\bm{D}^{H}\bm{f}$. The elements of $\bm{A}$ are chosen from i.i.d. standard normal distribution. We set
$ r= 2\|\bm{f}\|_{2} $, $ \sigma=r $, and the number of stages to $ T=10 $. Define the normalized reconstruction error as 
$\tfrac{\|\bm{f}-\hat{\bm{f}}\|_{2}}{\|\bm{f}\|_{2}}$. The results in right block of
Figure \ref{fig:SimFig1} are obtained by implementing Algorithms 
\ref{alg:AQ} and \ref{alg:AR} $100$ times and taking the average of 
the normalized reconstruction error. As it is clear from right block of Figure \ref{fig:SimFig1}, 
LP algorithm outperforms CP
by $ 2~dB $ on average. Our algorithm with few measurements
behaves slightly weaker than others. However, it seems that there 
is a phase transition behavior in our algorithm when the number of measurements 
increases. In fact, after a certain number of measurements, our proposed 
algorithm substantially outperforms both LP and CP 
(over $ 50~dB $ in the steady-state condition). 
\begin{figure}[t]
	\centering
	\includegraphics[scale=.35]{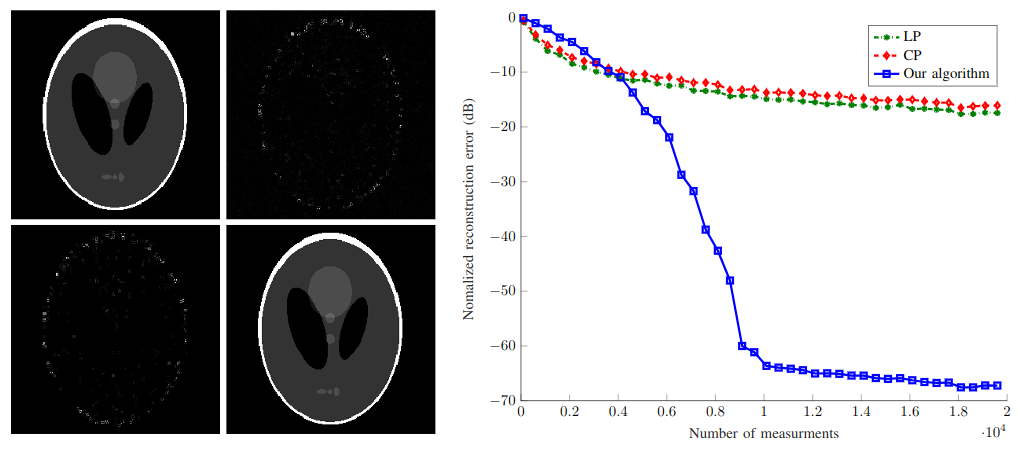}
	\caption{Left block: 
	Reconstruction of Shepp-Logan 
	phantom image of size $256\times 256$ where the picture is split
	into $64$ blocks of size $ 32 \times 32 $, each requires $ m=10^5 $ measurements. Top left image: Ground truth image ($\text{PSNR}= \infty~dB$). Top right image: Recovered image via LP ($\text{PSNR}= 12.02~dB$). Bottom left image: Recovered image via CP ($\text{PSNR}= 11.98~dB$). Bottom right image: Recovered image via our algorithm ($\text{PSNR}= 118.25~dB$).
		Right block: Reconstruction of a dictionary-sparse signal  $ \bm{x}\in \R^{50} $ in dictionary $ \bm{D}\in \R^{50\times 1000} $. The plot shows the reconstruction 
	 error averaged over $100$ Monte Carlo simulations for Algorithms LP, CP, and our proposed algorithm.
	}
	\label{fig:SimFig1}
\end{figure}
In the second experiment, we consider the Shepp-Logan phantom 
image as the ground-truth signal. Since 
the required number of measurements in one-bit CS is almost high (see e.g. \cite{baraniuk2017exponential} and \cite[Section 4.2]{needell2017weighted} ), we split the picture into multiple blocks 
of size $ 32\times 32 $ and process each block separately. We use redundant wavelet dictionary 
and $ 10^5 $ Gaussian measurements to recover each block. We evaluate the reconstruction quality of final 
result in terms of the peak signal to noise ratio (PSNR) given by $\text{PSNR}\left(\bm{X},\widehat{\bm{X}}\right) = 20 
\log_{10}\left(\tfrac{\|\bm{X}\|_{\infty}256}{\|\bm{X}-\widehat{\bm{X}}\|_{F}}\right),$
where $ \bm{X} $ and $ \widehat{\bm{X}} $ are the true and estimated images.
As shown by left block of Figure \ref{fig:SimFig1}, LP and CP algorithms 
clearly fail with a poor performance but as it is evident in the bottom right image (left block of Figure \ref{fig:SimFig1}), the output of our algorithm is almost similar to the desired picture.
\appendix
\section{Proof}
\label{appndx}
\subsection{Proof of Theorem \ref{thm.main}}\label{proof.theoremmain}
\begin{proof}
By induction law, we show that
\begin{align}
\label{eq:IndHpt}
\|\bm{f}-\bm{f}_{i}\|_{2}\leq \epsilon r2^{1-i}
\end{align}
holds with high probability for any $i\in\{1,\dots,T\}$. 
Consider the first step i.e. $i=1$ in Algorithm \ref{alg:AQ}. At this step, 
our initial estimate $\bm{f}_0$ is equal to $\bm{0}$. Thus, the output of 
Algorithm \ref{alg:HDTG} only contains the random part of high dimensional 
thresholds (step 2 of Algorithm \ref{alg:AQ}). Then, we obtain $\bm{f}_1$ by 
using steps 4-9 of Algorithm \ref{alg:AQ} (except that we assume $\sigma=r$ in 
Algorithm \ref{alg:AQ}). As a result, to verify $\|\bm{f}-\bm{f}_{1}\|_{2}\leq \epsilon r$, we use \cite[Theorem 8]{Baraniuk2017}.
Now suppose that the result \eqref{eq:IndHpt} holds in the $(i-1)$-th step, i.e,
\begin{align}
	\label{eq:IndHpt2}
	\|\bm{f}-\bm{f}_{i-1}\|_{2}\leq \epsilon r2^{2-i}.
\end{align}
Consider $i$-th stage of Algorithm \ref{alg:AQ} where the
high dimension thresholds and the measurements are obtained as
\begin{align}
	\label{eq:proof1}
	\bm{\varphi}^{\left(i\right)}\leftarrow \Phi(\bm{A}^{(i)},q,2^{1-i}r,\bm{f}_{i-1}),
\end{align}
\begin{align}
	\label{eq:proof2}
	\bm{y}^{\left(i\right)} = \text{sign}\left(\bm{A}^{(i)}\bm{f}-\bm{\varphi}^{\left(i\right)}\right).
\end{align}
By substituting \eqref{eq:proof1} in \eqref{eq:proof2}, we reach:
‌\begin{align}
	\label{eq:proof3}
	\bm{y}^{\left(i\right)} = \text{sign}\left(\bm{A}^{(i)}(\bm{f}-\bm{f}_{i-1})-\bm{\tau}^{(i)}\right).
\end{align}	
Since $\bm{f}$ is effective $s$-analysis-sparse and the output of Algorithm \ref{alg:AQ} at 
the $(i-1)$-th stage, i.e. $\bm{f}_{i-1}$ is effectively $s$-analysis-sparse, 
the signal $\bm{f}-\bm{f}_{i-1} $ is effectively $4s$-analysis-sparse. 
Now, we can apply \cite[Theorem 8]{Baraniuk2017} to this signal. In simple words, we set
\begin{align}
&\bm{f}\leftarrow \bm{f}-\bm{f}_{i-1}, r\leftarrow 2^{-i}r, \sigma\leftarrow 2^{-i}r, s\leftarrow 4s
\end{align}
in \cite[Theorem 8]{Baraniuk2017}. As a result, we shall have that
‌\begin{align}
\label{eq:proof4}
\norm{\bm{f}-\bm{f}_{i-1}-\bm{\Delta}_{i}}_{2} \leq \epsilon 2^{-i}r,
\end{align}
with probability at least $1- \gamma \exp(-c^{\prime}q)$. Now, suppose 
that \eqref{eq:proof4} occurs. Consider
\begin{align}
\label{eq:proof5}
\bm{f}_{\rm tmp}=\bm{f}_{i-1}+\bm{\Delta}_{i}.
\end{align}
After applying step 6 of Algorithm \ref{alg:AQ}, we obtain $\bm{f}_i$ that has the property
$$\|\bm{D}^{\rm H}\bm{f}_i\|_1\le \sqrt{s}r.$$
Then, we have
\begin{align}
\|\bm{f}-\bm{f}_i\|_2\le\|\bm{f}-\bm{f}_{\rm tmp}\|_2+\|\bm{f}_{\rm tmp}-\bm{f}_i\|_2.
\end{align}
Finally, by using \eqref{eq:proof4}, \eqref{eq:proof5}, and the fact that $\|\bm{f}_i-\bm{f}_{\rm tmp}\|_2\le \|\bm{f}-\bm{f}_{\rm tmp}\|_2$ (step 6),  the latter equation becomes
\begin{align}
\|\bm{f}-\bm{f}_i\|_2\le \epsilon 2^{1-i}r
\end{align}
Since we consider $T$ stage in Algorithm \ref{alg:AQ} and \ref{alg:AR}, we reach the error bound:
‌\begin{align}
	\label{eq:proof6}
	\norm{\bm{f}-\bm{f}_{T}}_{2} \leq \epsilon r 2^{1-T}.
\end{align}
\end{proof}

\ifCLASSOPTIONcaptionsoff
\newpage
\fi

\bibliographystyle{ieeetr}
\bibliography{HBReference}
\end{document}